\documentclass[pdflatex,sn-mathphys-num]{sn-jnl}

\usepackage{graphicx}%
\usepackage{multirow}%
\usepackage{amsmath,amssymb,amsfonts}%
\usepackage{amsthm}%
\usepackage{mathrsfs}%
\usepackage[title]{appendix}%
\usepackage{xcolor}%
\usepackage{textcomp}%
\usepackage{manyfoot}%
\usepackage{booktabs}%
\usepackage{caption}
\usepackage{float}
\usepackage{array,tabularx}

\def\h{\mathbf{h}}
\def\C{\mathcal{C}}
\def\h{\mathbf{h}}
\def\c{\mathbf{c}}
\def\Fq{\mathbb{F}_{q}}
\def\F{\mathbb{F}}

\def\u{\mathbf{u}}

\def\v{\mathbf{v}}
\def\e{\mathbf{e}}

\def\x{\mathbf{x}}

\def\G{\mathcal{G}}

\def\W{\mathcal{W}}

\def\0{\mathbf{0}}
\def\1{\mathbf{1}}
\def\V{\mathcal{V}}

\def\L{\mathcal{L}}

\numberwithin{equation}{section}
\newtheoremstyle{mystyle}{3pt}{3pt}{}{\parindent}{\bfseries}{.}{5mm}{}  
\theoremstyle{mystyle}

\newtheorem{Definition}{Definition}[section]

\newtheorem{Lemma}[Definition]{Lemma}
\newtheorem{Construction}[Definition]{Construction}
\newtheorem{Corollary}[Definition]{Corollary}
\newtheorem{Theorem}[Definition]{Theorem}
\newtheorem{Example}[Definition]{Example}
\newtheorem{Remark}[Definition]{Remark}
\newtheorem{Open problem}[Definition]{Open problem}

\begin{document}

\title[Article Title]{Constructions of locally repairable codes via concatenated codes\footnote{This research is supported by the National Key Research and Development Program of China (Grant No. 2022YFA1005000), the National Natural Science Foundation of China (Grant No. 62371259), the Fundamental Research Funds for the Central Universities of China (Nankai University), and the Nankai Zhide Foundation.}}


\author[1]{\fnm{Hengfeng} \sur{Jin}}\email{jinhengfeng@mail.nankai.edu.cn}

\author*[1]{\fnm{Fang-Wei} \sur{Fu}}\email{fwfu@nankai.edu.cn}

\affil[1]{\orgdiv{Chern Institute of Mathematics and LPMC}, \orgname{Nankai University}, \orgaddress{\city{Tianjin}, \postcode{300071}, \country{China}}}




\abstract{Locally repairable codes (LRCs) have attracted considerable attention owing to their pivotal role in distributed storage systems. Since binary linear LRCs can significantly reduce the complexity of both encoding and decoding, the construction of binary LRCs has drawn extensive research interest. In this paper, we construct binary LRCs with locality \(r=2\) and disjoint local repair groups via classical concatenation, where the outer codes are linear codes over \(\F_4\) and the inner code is the binary \([3,2,2]\) linear code. We present a systematic approach for selecting outer codes to obtain LRCs attaining specified bounds, and show that the weight distributions of the resulting binary LRCs are determined by those of the selected outer codes. Under certain sufficient conditions, we construct several classes of optimal binary LRCs, including those attaining Griesmer-type bounds and binary perfect LRCs. Furthermore, we sharpen the Johnson-like bounds of Ma and Ge for binary LRCs with locality \(r=2\) and disjoint local repair groups, and define nearly perfect binary LRCs and \(k\)-optimal binary LRCs with respect to the improved bounds. We show that nearly perfect binary LRCs can be constructed via concatenation when the outer codes are nearly perfect linear codes over \(\F_4\). Finally, we present several examples of such nearly perfect and \(k\)-optimal binary LRCs.}


\maketitle
\clearpage

\section{Introduction}\label{sec1}

Distributed storage systems, such as data centers, store data in a redundant form to ensure reliability in the event of node failures. Erasure codes are widely utilized in distributed storage systems. Locally repairable codes (LRCs), proposed by Gopalan \emph{et al.} \cite{Gopalan}, are a significant class of erasure codes widely utilized in distributed storage systems where they offer notable benefits such as low storage cost and efficient data recovery. LRCs are also used for coded distributed computing to encode distributed tasks in a linear matrix-vector multiplication problem, as demonstrated in \cite{fetrat}.

In this context, let $\C$ be an $[n,k,d]$ linear code over $\Fq$, which is a subspace of $\Fq^n$ with dimension $k$ and minimum distance $d$. Then $\C$ is called a locally repairable code with locality $r$ if any failed code symbol can be recovered by accessing at most $r$ other code symbols.
For efficient local repair, it is desirable to have $r$ much smaller than $k$. At the same time, larger values of $k$ and $d$ are desirable to achieve high storage efficiency and strong global fault tolerance, respectively. Indeed, there exists a tradeoff among these parameters. Gopalan \emph{et al.} \cite{Gopalan} proved that the minimum distance of an LRC is upper bounded by
\begin{equation}\label{s1}
    d \leq n-k-\left \lceil   \frac{k}{r}\right \rceil+2.
\end{equation}

Since then, numerous scholars have constructed various $d$-optimal LRCs achieving the Singleton-like bound \eqref{s1} (see \cite{Tamo1,Jin1,Hao1,Luo1,Cai2}).

Binary LRCs are of particular interest because binary linear codes enable encoding and linear erasure recovery using only XOR operations, thereby avoiding nontrivial finite-field multiplications. Hao \emph{et al.} \cite{Hao} showed that there are
only 5 classes of binary LRCs that meet the bound \eqref{s1}.

However, the Singleton-like bound does not take into account the size of the finite field. As a result,  this bound may not be tight for some binary LRCs. There has been extensive research on binary locally repairable codes (see \cite{Pengfei,Shah,Luo2021tcom,Wang,Ma}).

Cadambe and Mazumdar \cite{Cadambe} proposed the $C$-$M$ bound for LRCs, which is the first field-size-dependent bound, as follows:
\begin{equation}
    k\leq  \min_{\substack{\tau\in\mathbb Z_{\geq0}\\
    \tau\leq \min\{\lfloor n/(r+1)\rfloor,\lceil k/r\rceil-1\}}}
    \bigl[\tau r+k_{\mathrm{opt}}^{q}(n-\tau(r+1),d)\bigr],
\end{equation}
where $k_{\mathrm{opt}}^{q}(n' , d' )$ denotes the largest possible dimension of a $q$-ary linear code with length $n'$ and minimum distance $d'$.  Hao \emph{et al.} \cite{Hao} also derived several field-size-dependent bounds for LRCs based on parity-check matrices, including the Griesmer-like bound. In \cite{Gaoyun}, Gao and Fang constructed optimal locally repairable codes meeting the Griesmer-like bound by lengthening Reed–Muller codes.

In \cite{Wang}, Wang \emph{et al.}  proposed a sphere-packing approach for deriving an upper bound on the dimension $k$ of $[n, k, d]$ LRCs with locality $r$. Through specific comparisons, it was demonstrated that this bound outperforms the $C$-$M$ bound in certain cases. 
Building on \cite{Wang}, Fang \emph{et al.} \cite{Fang} introduced the concept of perfect LRCs and constructed several classes of such codes. They further constructed perfect LRCs via cyclic and constacyclic codes in \cite{FANG2023}. In \cite{Ma}, Ma and Ge established an explicit Johnson-like bound on the dimension \(k\) of binary LRCs with disjoint local repair groups, by an argument analogous to the one used to derive the Johnson bound for binary codes in \cite[Theorem 2.3.8]{Huffman}. They also obtained several families of such \(k\)-optimal LRCs attaining this bound.

In this paper, we construct binary LRCs with locality \(r=2\) and disjoint local repair groups via classical concatenation, where the outer codes are linear codes over \(\F_4\) and the inner code is the binary \([3,2,2]\) linear code. We present a systematic approach for selecting outer codes to obtain binary LRCs attaining prescribed bounds. The weight distributions of the resulting LRCs are determined by those of the selected outer codes. We also establish a minimum-distance criterion for general configurations of two-dimensional binary subspaces. Unlike concatenation with an $\F_4$-linear outer code, this criterion also applies to binary LRCs of odd dimension. Under certain sufficient conditions, we construct several classes of optimal binary LRCs, including binary LRCs attaining Griesmer-type bounds and binary perfect LRCs.

For $r=2$, we sharpen the Johnson-like bound of Ma and Ge for binary LRCs with disjoint local repair groups. Concatenation with the hexacode yields a binary $[18,6,8;2]$ LRC attaining the refined counting bound~\eqref{npbbb}, hence nearly perfect in the sense of Definition~\ref{near}. A code with the same parameters was given in \cite{Ma} to demonstrate equality in the dimension bound~\eqref{last}. We also construct binary $[33,14,8;2]$ and
$[36,12,12;2]$ LRCs attaining the refined dimension
bound~\eqref{las}, which improves the Ma--Ge upper
bound~\eqref{last} on $k$ from $15$ to $14$
and from $13$ to $12$, respectively.

The binary perfect LRCs obtained from the Hamming codes over \(\F_4\) via concatenation in Section~5 have parameters
\[
[2^{2s}-1,\tfrac{2}{3}(2^{2s}-1)-2s,6;2],\quad s\geq 2,
\]
and their weight distributions follow directly from those of the underlying Hamming codes over \(\F_4\). Binary LRCs  with the same parameters also appeared in Goparaju--Calderbank's construction \cite[Construction~2]{Goparaju}, Wang--Zhang--Lin's construction with \(t=1\) \cite[Construction~1]{Wang}, Fang et al.'s construction of perfect LRCs \cite[Theorem~4(i)]{Fang}, and Ma--Ge's construction \cite[Corollary~4.13]{Ma}.  Table~\ref{tab:comparison} compares the parameters of our binary LRCs
with those of codes in \cite{Goparaju,Wang,Fang,Ma}, highlighting
the attained bounds, weight-distribution information, and
tradeoffs between dimension and minimum distance.

\begin{table*}[htbp]
\centering
\caption{Representative binary LRCs with locality two and disjoint repair groups: parameter comparisons, attained bounds, and contributions of the present constructions.}
\label{tab:comparison}

\footnotesize
\setlength{\tabcolsep}{3pt}
\renewcommand{\arraystretch}{1.15}
\renewcommand{\tabularxcolumn}[1]{m{#1}}

\begin{tabularx}{\textwidth}{
    >{\centering\arraybackslash}m{0.55cm}|
    >{\centering\arraybackslash}m{2.25cm}|
    >{\centering\arraybackslash}m{2.25cm}|
    >{\raggedright\arraybackslash}m{2.50cm}|
    >{\raggedright\arraybackslash}X
}
\hline
$n$
& Known parameters
& Parameters obtained here
& Attained bounds
& Contribution and comparison \\
\hline

15
& $[15,6,6;2]$
& $[15,6,6;2]$
& Sphere-packing-like bounds
  \eqref{perf}, \eqref{perf1}
& Concatenation with the \([5,3,3]\) Hamming code over \(\F_4\) as the outer code yields these parameters and determines the weight distribution. \\

& $[15,4,8;2]$
& $[15,4,8;2]$
& Griesmer-type bound \eqref{gri};
  refined Johnson-like bound \eqref{las}
& Concatenation with a $[5,2,4]$ MDS outer code
  over $\F_4$ gives these parameters and determines
  the weight distribution. \\
\hline

18
& \shortstack[l]{$[18,7,6;2]$;\\ $[18,6,8;2]$}
& $[18,6,8;2]$
& Refined Johnson-like bound
  \eqref{npbbb}
& Concatenation with the hexacode gives a nearly
  perfect LRC with $A_0=1$, $A_8=45$, and $A_{12}=18$. \\
\hline

27
& $[27,13,6;2]$
& $[27,4,14;2]$
& Griesmer-type bound \eqref{gri}, attained at $\tau=0$  (classical Griesmer bound).
& A different dimension--distance tradeoff:
  $(k,d)=(4,14)$ versus $(13,6)$. \\
\hline

33
& $[33,16,6;2]$
& $[33,14,8;2]$
& Refined Johnson-like bound \eqref{las}
& At $(n,d)=(33,8)$, the code attains $k\leq14$
  from \eqref{las}, whereas \eqref{last} gives $k\leq15$. \\
\hline

36
& $[36,18,6;2]$
& $[36,12,12;2]$
& Refined Johnson-like bound \eqref{las}
& At $(n,d)=(36,12)$, the code attains $k\leq12$
  from \eqref{las}, whereas \eqref{last} gives $k\leq13$. \\
\hline

51
& $[51,28,6;2]$
& $[51,26,8;2]$
& Johnson-like bounds \eqref{last}, \eqref{las}
& A different dimension--distance tradeoff:
  $(k,d)=(26,8)$ versus $(28,6)$. \\

& $[51,28,6;2]$
& $[51,8,24;2]$
& Griesmer-type bound \eqref{gri}
& A different dimension--distance tradeoff:
  $(k,d)=(8,24)$ versus $(28,6)$. \\
\hline

129
& $[129,78,6;2]$
& $[129,72,10;2]$
& Sphere-packing-like bound \eqref{perf1}
& A different dimension--distance tradeoff:
  $(k,d)=(72,10)$ versus $(78,6)$. \\
\hline
\end{tabularx}

\par\smallskip
\noindent
\parbox{\textwidth}{%
\footnotesize\raggedright
All $d=6$ parameters in the second column are from
\cite[Corollary~4.13 and Table~2]{Ma};
$[15,4,8;2]$ and $[18,6,8;2]$ are from
\cite[Remark~5.10 and Example~5.9, respectively]{Ma}.
The fourth column lists bounds attained by the codes
in the third column.
}
\end{table*}

The rest of this paper is organized as follows. In Section~2, we introduce notation and basic concepts related to concatenated codes, LRCs, and projective geometry. In Section~3, we establish a subspace criterion for minimum distance and present the concatenated construction together with its weight distribution.  In Section~4, we construct binary LRCs attaining the Griesmer-type bounds. In Section~5, we construct binary perfect LRCs and derive their weight distributions. In Section~6, we sharpen the Johnson-like bound of Ma and Ge for binary LRCs with locality \(r=2\) and disjoint local repair groups, and define nearly perfect binary LRCs and \(k\)-optimal binary LRCs with respect to the improved bounds. We show that nearly perfect binary LRCs can be constructed via concatenation when the outer codes are nearly perfect linear codes over \(\F_4\), and present several examples of such nearly perfect and \(k\)-optimal binary LRCs. In Section~7, we conclude the paper.

\section{Preliminaries}
\label{PRELIMINARIES}
In this section, we introduce notation and basic concepts related to concatenated codes, LRCs, and projective geometry.

Let $[n]$ denote the set $\{1, 2,\dots, n\}$ and $[a:b]$ denote the set $\{a, a+1,\dots, b\}$. Let $q$ be a prime power and  $\F_q$ be a finite field with $q$ elements.  Let $\Fq^\ast= \Fq \setminus \{0\}$ and $\F_q^n$ denote the vector space of dimension $n$ over $\F_q$.

For a vector $\v = (v_1, v_2,\dots, v_n)\in \Fq^n$, the support of $\v$ is
defined as $\mathrm{supp}(\v) \triangleq \{i \in [n]:  v_i\neq 0\}$.  The Hamming weight $\mathrm{wt}(\v)$ of $\v$ is defined as the size of $\mathrm{supp}(\v)$, i.e., $\mathrm{wt}(\v) = |\mathrm{supp}(\v)|$.   The Hamming distance between two vectors $\u$ and $\v$ in $\Fq^n$ is defined as $d(\u,\v)\triangleq \mathrm{wt}(\u-\v)$.

Let $\v^T$ denote the transpose of the vector $\v$, and $\0_n$ and $\1_n$ denote the all-zero and all-one vectors of length $n$ over $\F_q$, respectively. Let $\operatorname{span}\{\v,\v'\}$ denote the subspace of $\F_q^n$ spanned by $\{\v,\v'\}$.

Let $\C$ be an $[n,k,d]$ linear code  over $\Fq$, which is a subspace of $\Fq^n$ with dimension $k$ and minimum distance $d$, where $d=\min\{d(\c_1,\c_2)|\c_1\neq \c_2,\c_1,\c_2 \in \C \}$. For $\v\in \F_q^n$, $d(\v, \C)\triangleq\min\{d(\v,\c)| \c \in \C\}$. The generator matrix of $\C$ is any $k\times n$ matrix $G$ over $\Fq$ whose rows form a basis for $\C$. The parity-check matrix of  $\C$ is an $(n-k) \times n$ matrix $H$ over $\Fq$ with $\mathrm{rank}(H)=n-k$ such that for every $\c\in \C$, $H\c^T=\0_{n-k}^T$. 

The dual code of $\C$ is defined by
$\C^\perp
=
\{
\mathbf{x}\in\F_q^n:
\sum_{i=1}^{n}x_i c_i=0
\text{ for every }\mathbf{c}\in\C
\}$.
Then $\C^\perp$ is a linear $[n,n-k]$ code over
$\F_q$, and any generator matrix of $\C^\perp$
is a parity-check matrix of $\C$.

 It is well known that  the  minimum distance of $\C$ is $d$ if and only if any $d-1$ columns of $H$ are linearly independent and there exist $d$ columns of $H$ which are linearly dependent.

For an $[n, k, d]$ linear code over $\F_q$, the minimum distance $d$ satisfies the Singleton bound $d \leq n-k+1$.

\begin{Definition}
 A linear code $\C \subseteq \mathbb{F}_q^n$ of length $n$, dimension $k$, and minimum  distance $d$ is called a maximum distance separable (MDS) code if it attains the Singleton bound with equality, i.e., $d = n - k + 1$. 
\end{Definition}

We generally denote the weight distribution of a code $\C$ by $A_0(\C)$, $A_1(\C)$, $\dots$, $ A_n(\C)$, where $A_i(\C)$ is the number of codewords of weight $i$.

Concatenated codes \cite[Section~1.2]{concatenated} form a fundamental class of error-correcting codes in coding theory
 that combine two or more simpler codes to achieve high levels of data reliability with manageable computational complexity.

We describe the construction of a concatenated code over $\F_q$. Let $\C_{\mathrm{in}}$ be an $[n, k, d]$ inner linear code over $\F_q$, and $\C$ be an $[N,K,D]$ outer linear code  over $\F_{q^k}$. Let $f$ be an $\F_q$-linear isomorphism from $\F_{q^k}$ to $\C_{\mathrm{in}}$. Then the corresponding concatenated code is \[\C'=\{(f(c_1)| f(c_2)| \dots | f(c_N))\mid \c=(c_1,c_2,\dots, c_N)\in \C\}.\]

It is straightforward to verify that $\C'$ is an $[nN, kK, \geq dD]$ linear code over $\F_q$. More details can be found in \cite{concatenated, Roth_2006}.

Next, we give the definition of locally repairable codes.
\begin{Definition}\label{df1}
The  $i$-th code symbol $c_{i}$ $(1 \le i \le n)$ of a linear code $\C$ is said to have  locality $r$ if there exists a subset $S_i \subset [n] $ with $|S_i|\leq r+1$ such that $i\in S_i$ and $c_{i}$ can be recovered from the symbols \(\{c_j:j\in S_i\setminus\{i\}\}\); equivalently, $c_i$ is a linear combination of these symbols. The set $S_i$ is called a repair group for $c_i$. 

\end{Definition}

An $[n,k,d]$ linear code $\C$ is said to have locality $r$ if all its symbols have locality $r$; such a code is called an $[n,k,d;r]$-LRC.

\begin{Lemma}[\cite{Hao}]\label{lem1}
    A linear code $\C$ has locality $r$ if and only if for every $1\leq i \leq n$, there exists a codeword $\v$ in $\C^{\bot}$ whose support set $\mathrm{supp}(\v)$ contains $i$ and the size of $\mathrm{supp}(\v)$ is at most $r+1$.
\end{Lemma}

Next, we provide an overview of the key concepts of projective geometry that are relevant to our results.

\begin{Definition}
  Two nonzero vectors $\u$ and $\v$ in $\Fq^{n+1}$ are equivalent if $\u=\lambda\v$ for some $\lambda\in\Fq^\ast$. The set of equivalence classes is the $n$-dimensional projective space over $\Fq$, denoted by $PG(n,q)$.
\end{Definition}

The elements of $PG(n,q)$ are called points. For a nonzero vector $A\in\Fq^{n+1}$, let $P(A)$ denote the point containing $A$.

A subspace of dimension $m$, or an $m$-space, of $PG(n,q)$ is a set of points whose representing vectors, together with the origin $\0$, form a subspace of dimension $m+1$ of $\Fq^{n+1}$. Each $m$-space contains $\frac{q^{m+1}-1}{q-1}$ points. For convenience, we refer to a 1-space as a line.

A set of points in $PG(n,q)$ is said to be collinear if all the points lie on a single line of the projective space.


\section{ Construction of LRCs via concatenated codes }

Throughout this paper, we focus on binary LRCs with disjoint local repair groups.

Following \cite[Definition~2.11]{Ma}, an $[n,k,d;2]$
binary LRC $\C$ with disjoint local repair groups has
$3\mid n$ and, up to a permutation of coordinates,
admits a parity-check matrix of the following form:
\begin{equation}\label{H1}
\begin{aligned}
H
&=(\h_1,\h_2,\dots,\h_n)=(\h_1^0,\h_1^1,\h_1^2,
   \h_2^0,\h_2^1,\h_2^2,\dots,
   \h_{\ell}^0,\h_{\ell}^1,\h_{\ell}^2)\\
&=
\left(
\begin{array}{ccc|ccc|c|ccc}
1&1&1 & & & & & & &\\
 & & &1&1&1 & & & &\\
 & & & & & &\ddots & & &\\
 & & & & & & &1&1&1\\
\hline
\mathbf{0}_{u}^T &\e_{1,1}^T &\e_{1,2}^T
&\mathbf{0}_{u}^T &\e_{2,1}^T &\e_{2,2}^T
&\dots
&\mathbf{0}_{u}^T &\e_{\ell,1}^T &\e_{\ell,2}^T
\end{array}
\right),
\end{aligned}
\end{equation}
where $n=3\ell$, $n-k=\ell+u$,
$\h_m^T\in\F_2^{\ell+u}$ for $m\in[n]$,
$(\h_i^j)^T\in\F_2^{\ell+u}$ for
$i\in[\ell]$ and $j=0,1,2$, and
$\e_{i,j'}\in\F_2^u$ for
$i\in[\ell]$ and $j'=1,2$.
The first $\ell$ rows give $\ell$ disjoint local repair
groups of size three, with group $i$ corresponding to
the columns $\{\h_i^0,\h_i^1,\h_i^2\}$, and ensure the locality
$r=2$.

The zero vector in the last $u$ entries of the first
column of each group is a normalization. Indeed, for
each $i\in[\ell]$, adding the $i$th row to any of the
last $u$ rows whose entry in that column is one makes
all these entries zero. These elementary row operations
preserve the code and do not affect  other groups,
so the form in \eqref{H1} entails no loss of generality.

For each $i\in[\ell]$, let
$\W_i=\operatorname{span}\{\e_{i,1},\e_{i,2}\}$.
The following theorem provides a minimum-distance criterion for binary LRCs whose parity-check matrix has the form \eqref{H1}. The criterion is expressed in terms of the dimensions of sums of these subspaces.

\begin{Theorem}\label{thm28}
Let $\C$ be an $[n,k,d;2]$ binary LRC with the parity-check matrix $H$ given in \eqref{H1}. Assume that $\dim(\W_i)=2$ for every $i\in[\ell]$, and let $1\leq t<\ell$ with $u\geq 2t$. Then $d\geq 2t+2$ if and only if
\[
 \dim\left(\sum_{j=1}^{t}\W_{m_j}\right)=2t
\]
for all distinct $m_1,m_2,\dots,m_t\in[\ell]$. If this dimension condition holds for every such $t$-subset of $[\ell]$ and, in addition, there exist distinct indices $a_1,a_2,\dots,a_t,a_{t+1}\in[\ell]$ such that
\[
 \dim\left(\sum_{j=1}^{t+1}\W_{a_j}\right)<2t+2,
\]
then $d=2t+2$.
\end{Theorem}
\begin{proof}
The condition $\dim(\sum_{j=1}^{t}\W_{m_j})=2t$
is equivalent to the linear independence of
$\e_{m_1,1},\e_{m_1,2},\dots,\e_{m_t,1},\e_{m_t,2}$,
whereas $d\geq 2t+2$ is equivalent to the linear
independence of any $2t+1$ columns of $H$.

Consider a nontrivial relation among $2t+1$ distinct
columns of $H$ in \eqref{H1}:
\begin{equation}\label{D2}
\sum_{j=1}^{2t+1}\lambda_j\h_{i_j}
=\mathbf{0}_{\ell+u},
\qquad
\lambda_j\in\F_2,\quad \text{not all zero}.
\end{equation}
The first $\ell$ rows require an even number of columns
with nonzero coefficients in each repair group.
Thus, the columns with nonzero coefficients come from
$s$ repair groups, say $m_1,\dots,m_s$, with exactly
two columns from each group and $1\leq s\leq t$.
The last $u$ rows give
\[
\sum_{j=1}^{s}
\bigl(\alpha_{j,1}\e_{m_j,1}
+\alpha_{j,2}\e_{m_j,2}\bigr)=\mathbf{0}_u,
\qquad
(\alpha_{j,1},\alpha_{j,2})\neq(0,0).
\]
Within group \(m_j\), the coefficient pairs
\((\alpha_{j,1},\alpha_{j,2})=(1,0),(0,1),(1,1)\)
correspond, respectively, to the column pairs
\((\h_{m_j}^{0},\h_{m_j}^{1})\),
\((\h_{m_j}^{0},\h_{m_j}^{2})\), and
\((\h_{m_j}^{1},\h_{m_j}^{2})\).
Conversely, for each \(j\), choose the column pair corresponding to
\((\alpha_{j,1},\alpha_{j,2})\) under the above correspondence.
The corresponding \(2s\) columns of \(H\) sum to zero.
Since any set of fewer than $t$ groups can be extended
to $t$ groups, the two independence conditions
are equivalent.

Finally, suppose that the dimension condition holds
for every set of $t$ groups. Choose distinct indices $a_1,\dots,a_{t+1}\in[\ell]$ such that
$\dim(\sum_{j=1}^{t+1}\W_{a_j})<2t+2$.
The $2(t+1)$ vectors $\e_{a_j,1},\e_{a_j,2}$
for $1\leq j\leq t+1$ are therefore linearly dependent.
Hence there are coefficients $\beta_{j,1},\beta_{j,2}\in\F_2$,
not all zero, satisfying
\begin{equation}\label{eq:subspace-dependence}
\sum_{j=1}^{t+1}
\bigl(\beta_{j,1}\e_{a_j,1}
+\beta_{j,2}\e_{a_j,2}\bigr)=\mathbf{0}_u.
\end{equation}
Every pair $(\beta_{j,1},\beta_{j,2})$ must be nonzero.
Otherwise, \eqref{eq:subspace-dependence} would give
a nontrivial relation involving at most $t$ groups,
contradicting the assumption that
$\dim(\sum_{j=1}^{t}\W_{m_j})=2t$
for any $t$ distinct indices $m_1,\dots,m_t\in[\ell]$.

By the form of $H$ in \eqref{H1},
\eqref{eq:subspace-dependence} gives the column relation
\begin{equation}\label{eq:column-dependence}
\sum_{j=1}^{t+1}
\bigl[
(\beta_{j,1}+\beta_{j,2})\h_{a_j}^{0}
+\beta_{j,1}\h_{a_j}^{1}
+\beta_{j,2}\h_{a_j}^{2}
\bigr]
=\mathbf{0}_{\ell+u}.
\end{equation}
Since each pair $(\beta_{j,1},\beta_{j,2})$ is nonzero,
\eqref{eq:column-dependence} has exactly two nonzero
coefficients in each of the $t+1$ groups.
 This implies that $\C$ has a codeword of weight $2t+2$. Thus $d\leq2t+2$.
Together with $d\geq2t+2$, this yields $d=2t+2$.
\end{proof}

We next present constructions of LRCs via concatenated codes using a parity-check matrix approach. We begin by describing a parity-check matrix for a classical concatenated code, and then show how, for \(q=r=2\), the columns of the parity-check matrix of the outer code over \(\F_4\) determine the binary subspaces in Theorem~\ref{thm28}.

Let $\C_{\mathrm{in}}$ be a linear $[n, k]$ code over $ \F_q$ and let $\mathcal B=(\beta_1,\beta_2,\dots,\beta_k)$ be
an ordered basis of $\F_{q^k}$ over $\F_q$. Fix a $k\times n$ generator matrix $G_{\mathrm{in}}$ of $\C_{\mathrm{in}}$, and define the concatenated code $\C'$ over $\F_q$ by a linear $[N,K]$ outer code $\C$ over $\F_{q^k}$ and the $\F_q$-linear isomorphism $f: \F_{q^k} \rightarrow \C_{\mathrm{in}}$, where for every vector $\x \in \F_q^{k}$,
\begin{equation*}
    f(\mathcal B\x^T) = \x G_{\mathrm{in}}.
\end{equation*}

For each element $\alpha \in \mathbb{F}_{q^k}$, define the linear map $\psi_\alpha : \mathbb{F}_{q^k} \to \mathbb{F}_{q^k}$ by $\psi_\alpha(x) = \alpha x$. With respect to the ordered basis $\mathcal B=(\beta_1,\dots,\beta_k)$, let $L(\alpha) = (l_{i,j}(\alpha))_{i,j \in [k]}$ be the $k \times k$ matrix over $\mathbb{F}_q$ that represents $\psi_\alpha$. That is, for every vector $\mathbf{x} \in \mathbb{F}_q^k$,
\begin{equation*}
    \psi_{\alpha} (\mathcal B\x^T)=\mathcal B L(\alpha)\x^T,
\end{equation*}
where $\beta_j\alpha=\sum_{i=1}^k\beta_i l_{i,j}(\alpha).$

Let $Q$ be a $k \times n$ matrix over $\F_q$ that satisfies $QG_{\mathrm{in}}^T = I$, where $I$ is the $k \times k$ identity matrix.
Let $H_{\mathrm{in}}$ be an $(n-k) \times n$ parity-check matrix of $\C_{\mathrm{in}}$ over $\F_q$ and $H= (h_{i,j})_{i\in [N-K],j\in [N]}$ be an $(N-K)\times N$ parity-check matrix of $\C$ over $\F_{q^k}$. In \cite[p. 382]{Roth_2006}, the $(nN-kK)\times nN$ parity-check matrix of the concatenated code $\C'$ is given by
\begin{equation}\label{contH}
 \begin{aligned}
         H'&= \left(\begin{array}{c|c|c|c}
		H_{\mathrm{in}}& & &\\
		&H_{\mathrm{in}}& & \\
        &&\ddots&\\
        & & &H_{\mathrm{in}}\\
        \hline
        L(h_{1,1})Q &L(h_{1,2})Q &\cdots& L(h_{1,N})Q\\
	   L(h_{2,1})Q &L(h_{2,2})Q &\cdots& L(h_{2,N})Q\\
        \vdots &\vdots &\vdots& \vdots\\
         L(h_{N-K,1})Q &L(h_{N-K,2})Q &\cdots& L(h_{N-K,N})Q     
	\end{array}
    \right).
         \end{aligned}
\end{equation}

 In particular, let \(\C_{\mathrm{in}}\) be the \([r+1,r,2]\) code over \(\F_q\) whose parity-check matrix is \(H_{\mathrm{in}}=\1_{r+1}\), and let \(\C\) be an \([N,K,D]\) code over \(\F_{q^r}\). By Lemma~\ref{lem1}, together with the form of the first \(N\) rows of the parity-check matrix \(H'\) (see \cite[Equation~(3)]{Hao}), it follows that \(\C'\) is an \([(r+1)N,rK,\ge2D;r]\) LRC over \(\F_q\) with disjoint local repair groups.

Concatenation with an $[r+1,r,2]$ inner code is classical. Cadambe and Mazumdar \cite[Section~VI-A, especially Theorem~2]{Cadambe} used random linear outer codes to prove asymptotic achievability; they also discussed an extended Reed--Solomon outer code immediately before that theorem.

We now specialize to \(q=r=2\) and illustrate the explicit connection with Theorem~\ref{thm28}.

Write $\F_4=\F_2(w)$, where $w^2+w+1=0$,
and fix the ordered basis $(1,w)$ of $\F_4$ over $\F_2$.
Choose
\[
G_{\mathrm{in}}=
\begin{pmatrix}
1&1&0\\
1&0&1
\end{pmatrix},
\qquad
Q=
\begin{pmatrix}
0&1&0\\
0&0&1
\end{pmatrix},
\qquad
QG_{\mathrm{in}}^T=I_2.
\]
Let $g:\F_4\longrightarrow\F_2^2$ be the
coordinate map with respect to $(1,w)$:
\[
g(0)=(0,0),\qquad
g(1)=(1,0),\qquad
g(w)=(0,1),\qquad
g(w^2)=(1,1).
\]
The inner encoding is
$f(\alpha)=g(\alpha)G_{\mathrm{in}}$, so
\[
f(0)=(0,0,0),\qquad
f(1)=(1,1,0),\qquad
f(w)=(1,0,1),\qquad
f(w^2)=(0,1,1).
\]
With respect to the same ordered basis,
the multiplication matrices are
\[
L(0)=
\begin{pmatrix}0&0\\0&0\end{pmatrix},
\qquad
L(1)=
\begin{pmatrix}1&0\\0&1\end{pmatrix},
\qquad
L(w)=
\begin{pmatrix}0&1\\1&1\end{pmatrix},
\qquad
L(w^2)=
\begin{pmatrix}1&1\\1&0\end{pmatrix}.
\]

Let $\C$ be a linear $[n_1,k_1,d_1]_4$
outer code with $d_1\geq2$, put $s=n_1-k_1$,
and let
\[
H_{\C}=(\boldsymbol b_1,\dots,\boldsymbol b_{n_1})
=(h_{a,i})_{a\in[s],\,i\in[n_1]}
\]
be a full-rank parity-check matrix, where
$\boldsymbol b_i^T\in\F_4^s$.
Extend \(g\) coordinatewise to obtain the \(\F_2\)-linear isomorphism
\[
\G:\F_4^s\longrightarrow\F_2^{2s},
\qquad
\G(x_1,\dots,x_s)=(g(x_1),\dots,g(x_s)).
\]
For each entry $h_{a,i}$, we have
\[
L(h_{a,i})Q
=
\bigl(
\mathbf{0}_2^T
\mid g(h_{a,i})^T
\mid g(wh_{a,i})^T
\bigr).
\]
Consequently, \eqref{contH} has precisely the
form \eqref{H1}, with $\ell=n_1$, $u=2s$, and
\[
\e_{i,1}=\G(\boldsymbol b_i^T),
\qquad
\e_{i,2}=\G(w\boldsymbol b_i^T),
\qquad
\W_i=
\G\bigl(
\operatorname{span}_{\F_4}\{\boldsymbol b_i^T\}
\bigr).
\]
Since $d_1\geq2$, every column $\boldsymbol b_i$
is nonzero. Hence $\dim_{\F_2}(\W_i)=2$,
and $\e_{i,1},\e_{i,2}$ form a basis of $\W_i$. To apply Theorem~\ref{thm28}, we next relate the dimensions of the sums of these subspaces to the ranks of the corresponding column sets of \(H_{\C}\).

For any nonempty $I\subseteq[n_1]$, let
$H_{\C,I}$ denote the submatrix of $H_{\C}$
formed by the columns indexed by $I$.
Since
$\sum_{i\in I}\W_i
=
\G\bigl(
\operatorname{span}_{\F_4}
\{\boldsymbol b_i^T:i\in I\}
\bigr)$,
we obtain
$\dim_{\F_2}\left(\sum_{i\in I}\W_i\right)
=
2\operatorname{rank}_{\F_4}(H_{\C,I})$.

Thus, the dimension condition in Theorem~\ref{thm28} is equivalent to the linear independence over \(\F_4\) of the corresponding columns of \(H_{\C}\).
In particular, for $d_1\geq2$, every $d_1-1$ columns
of $H_{\C}$ are linearly independent, whereas some
$d_1$ columns are linearly dependent.
Taking $t=d_1-1$, the inequalities
$d_1-1\leq s$ and $d_1-1<n_1$ ensure that
$u\geq2t$ and $1\leq t<\ell$.
Theorem~\ref{thm28} therefore gives
$d(\C')=2d_1$.

The assumption $d_1\geq2$ is needed only for the preceding
application of Theorem~\ref{thm28}.
When $d_1=1$, a zero column of $H_{\C}$ gives a
zero-dimensional subspace, so the theorem does not apply.
The concatenation remains valid, however, as the direct
weight computation below also covers this case.
We therefore state the construction for all $d_1\geq1$.

\begin{Construction}\label{constr}
Fix the binary $[3,2,2]$ inner code
$\C_{\mathrm{in}}$ with parity-check matrix
$H_{\mathrm{in}}=(1,1,1)$ and the encoding $f$
defined above. Let the outer code $\C$ be an $[n_1,k_1,d_1]$ linear code over $\F_4$, where $d_1\geq 1$. Then the concatenated code $\C'$ is an $[n=3n_1,k=2k_1,d=2d_1;r=2]$ binary LRC with $n_1$ disjoint local repair groups.
\end{Construction}

\begin{proof}
The length, dimension, and locality follow directly from the general concatenation framework described above. It remains to verify the exact minimum distance. Every nonzero word of the binary \([3,2,2]\) inner code has weight two. Thus, for every \(\c\in\C\) and its concatenated image \(\c'\), we have
\[
\operatorname{wt}(\c')=2\operatorname{wt}(\c).
\]
An outer codeword of weight \(d_1\) produces a concatenated codeword of weight \(2d_1\), so \(d(\C')\le 2d_1\). Together with the lower bound \(d(\C')\ge 2d_1\), this yields \(d(\C')=2d_1\).
\end{proof}

Using the same argument, we can determine the entire weight distribution of the concatenated code. Weight distributions are useful for analyzing the performance of erasure codes, as discussed in~\cite{Shen}.

\begin{Corollary}\label{weight}
Let \(\C\) and \(\C'\) be as in Construction~\ref{constr}. Their weight distributions satisfy
\[
A_{2j}(\C')=A_j(\C),
\qquad 0\leq j\leq n_1,
\]
and \(A_i(\C')=0\) for every other index \(0\leq i\leq 3n_1\). In particular, \(A_0(\C')=A_0(\C)=1\), and \(\C'\) has no codewords of odd weight or of weight greater than \(2n_1\).
\end{Corollary}

For $d_1\geq2$, Construction~\ref{constr} realizes
a special case of Theorem~\ref{thm28}.
This theorem also applies to odd-dimensional binary LRCs,
which cannot arise from Construction~\ref{constr}
because its dimension is always $2k_1$.
For example, the known $[18,7,6;2]$ binary LRC from
\cite{Ma} 
has $\ell=6$ and $u=5$.
In the normalized form \eqref{H1}, its minimum distance
six ensures that each $\W_i$ has dimension two,
so Theorem~\ref{thm28} applies with $t=2$.

\begin{Remark}
With $\G$ fixed, the identity $w^2=w+1$ shows that
replacing $\G(w\boldsymbol b_i^T)$ by
$\G(w^2\boldsymbol b_i^T)$ in the generating pair
leaves $\W_i$ unchanged.
Changing the $\F_2$-basis of $\F_4$ applies a common
invertible binary linear transformation and therefore
preserves the dimensions of all subspace sums.

Moreover, every $\F_2$-linear isomorphism
$f:\F_4\to\C_{\mathrm{in}}$ maps each nonzero symbol
to an inner codeword of weight two.
Thus, the distance and weight-distribution correspondence
in Construction~\ref{constr} and Corollary~\ref{weight}
is independent of the choice of basis and inner encoding.
\end{Remark}

\begin{Remark}[The case $r>2$]\label{rem:general-r}
In the classical concatenation described above,
fixing $G_{\mathrm{in}}$ and choosing an $\F_2$-basis
$\mathcal B$ of $\F_{2^r}$ determines the inner encoding
through $f(\mathcal B\x^T)=\x G_{\mathrm{in}}$.
Changing $\mathcal B$ may assign a different inner
codeword to an outer symbol.
For $r>2$, the binary $[r+1,r,2]$ inner code has
nonzero codewords of different weights, so this choice
may affect the weights of concatenated codewords.

The Hamming weight of an outer codeword records only the number of its nonzero symbols, whereas different nonzero symbols may be mapped by \(f\) into inner codewords of different weights. Thus, the weight distribution of the outer code alone does not, in general, determine either the weight distribution or the minimum distance of the concatenated code.

Let the outer code \(\C\) have parameters \([n_1,k_1,d_1]\) over \(\F_{2^r}\). Every nonzero outer codeword has at least \(d_1\) nonzero symbols, each contributing at least two to the weight of the corresponding concatenated word. Therefore, \(d(\C')\ge 2d_1\). Equality holds precisely when there exists a minimum-weight outer codeword all of whose nonzero symbols are mapped to inner codewords of weight two.
\end{Remark}

\section{Construction of binary LRCs attaining Griesmer-type bounds}
In this section, we demonstrate that certain classes of linear codes over $\F_4$ attaining the classical Griesmer bound can be used to construct binary LRCs attaining the Griesmer-type  bound, including its classical boundary case.

The classical Griesmer bound for an $[n, k, d]$ linear code over $\F_q$ is
\begin{equation}\label{classical-gri}
    n\geq \sum_{i=0}^{k-1}\left \lceil   \frac{d}{q^i}\right \rceil.
\end{equation}

Hao \emph{et al.} \cite{Hao} derived a Griesmer-like bound for $[n,k,d;r]$ LRCs over $\F_q$ with $1\leq\tau\leq\lceil k/r\rceil-1$. Including the classical Griesmer bound as the term $\tau=0$, we use the following form:
\begin{equation}\label{gri}
    n \geq \max_{0\leq \tau\leq \left\lceil k/r\right\rceil-1}
    \left\{\tau(r+1)+\sum_{i=0}^{k-r\tau-1}\left\lceil\frac{d}{q^i}\right\rceil\right\}.
\end{equation}
Here $\tau$ is an integer. The term $\tau=0$ is precisely the classical Griesmer bound \eqref{classical-gri}, whereas the terms $\tau\geq1$ are those of the bound in \cite{Hao}.

\begin{samepage}
\begin{Theorem}\label{thm:mds}
 In Construction \ref{constr}, let the outer code $\C$ be an $[n_1,k_1,d_1=n_1-k_1+1]$ MDS code over $\F_4$, where $1\leq k_1<n_1$. Then the concatenated code $\C'$ is an $[n=3n_1,k=2k_1,d=2(n_1-k_1+1);2]$ binary LRC that attains the Griesmer-type  bound \eqref{gri}.
\end{Theorem}
\end{samepage}

\begin{proof}
  Set \(\tau_1=k_1-1\), which is admissible in \eqref{gri}. Then
\[
3\tau_1+\sum_{i=0}^{2k_1-2\tau_1-1}
\left\lceil\frac{2d_1}{2^i}\right\rceil
=3(k_1-1)+3d_1=3n_1=n.
\]
Since this admissible choice gives \(n\) and the maximum in \eqref{gri} is at most \(n\), the maximum equals \(n\).
\end{proof}

The trivial \([n_1,n_1,1]\) MDS code over \(\F_4\) yields an optimal \([3n_1,2n_1,2;2]\) binary LRC attaining the Singleton-like bound \eqref{s1}, since
$3n_1-2n_1-\left\lceil \frac{2n_1}{2}\right\rceil+2=2$.

Recall that an \([n_1,k_1,d_1=n_1-k_1+1]\) linear MDS code is called nontrivial if \(2\le k_1\le n_1-2\). The possible parameter sets of nontrivial MDS codes over \(\F_4\) can be determined by shortening and duality. Shortening such an MDS code in \(k_1-2\) coordinates gives an \([n_1-k_1+2,2,d_1]\) MDS code over \(\F_4\). Since any two columns of its generator matrix are linearly independent, these columns represent distinct points of \(PG(1,4)\), which has five points; hence \(n_1-k_1\le3\). Applying the same argument to the dual MDS code gives \(k_1\le3\). Together with nontriviality, this yields \(k_1,n_1-k_1\in\{2,3\}\). Thus the only possible parameter sets of nontrivial MDS codes over \(\F_4\) are
$[4,2,3], [5,2,4], [5,3,3],  [6,3,4]$.
The first three are realized by the extended Reed--Solomon codes or their duals, and the last by the hexacode in Example~\ref{hexacode}. This argument determines the code parameters, not the equivalence classes. Table~\ref{tab:mds_lrc} lists these parameter sets together with those of the corresponding binary LRCs.

\begin{table}[ht]
\centering
\caption{Possible parameters of nontrivial linear MDS codes over $\mathbb{F}_4$ and corresponding optimal binary LRCs}
\label{tab:mds_lrc}
\begin{tabular}{ll}
\hline
Quaternary linear MDS parameters & Corresponding optimal binary LRC parameters  \\
\hline
$[n_1=4,k_1=2,d_1=3]$ & $[n=12,k=4,d=6;r=2]$ \\
$[n_1=5,k_1=2,d_1=4]$ & $[n=15,k=4,d=8;r=2]$ \\
$[n_1=5,k_1=3,d_1=3]$ & $[n=15,k=6,d=6;r=2]$ \\
$[n_1=6,k_1=3,d_1=4]$ & $[n=18,k=6,d=8;r=2]$ \\
\hline
\end{tabular}
\end{table}

We next give a sufficient condition for a quaternary
Griesmer code to yield a binary LRC attaining
the Griesmer-type bound~\eqref{gri}.
\begin{Theorem}\label{t42}
   Let \(4^{\nu-1}<d_1=a\cdot4^{\nu-1}-\mu\leq 4^\nu\), where \(\nu>1\) is a positive integer, \(a\in\{2,3,4\}\), and \(\mu=\sum_{i=0}^{\nu-2}\epsilon_i4^i\) with each \(\epsilon_i\in\{0,1\}\). In Construction~\ref{constr}, let the outer code \(\C\) be a Griesmer code over \(\F_4\) with parameters \([n_1,k_1,d_1]\) that attains the classical Griesmer bound \eqref{classical-gri}. Then the concatenated code \(\C'\) is an \([n=3n_1,k=2k_1,d=2d_1;2]\) binary LRC attaining the Griesmer-type bound \eqref{gri}. If \(k_1\leq\nu\), equality is attained at \(\tau=0\); consequently, \(\C'\) also attains the classical Griesmer bound.
\end{Theorem}

\begin{proof}
Set \(\tau_3=\max\{k_1-\nu,0\}\) and \(b=k_1-\tau_3=\min\{k_1,\nu\}\). Since the outer code attains the Griesmer bound and $\lceil d_1/4^i\rceil=1$ for $i\geq\nu$, we have
\[
 n_1=\tau_3+\sum_{i=0}^{b-1}\left\lceil\frac{d_1}{4^i}\right\rceil.
\]
Since \(\epsilon_j\in\{0,1\}\), the remainder of \(\mu\) modulo \(4^i\) is less than \(4^i/2\) for \(0\le i\le \nu-1\). Hence the fractional part of \(\mu/4^i\) is less than \(1/2\), so
\[
2\left\lfloor\frac{\mu}{4^i}\right\rfloor
=\left\lfloor\frac{2\mu}{4^i}\right\rfloor.
\]
Together with \(d_1=a4^{\nu-1}-\mu\), this yields
\begin{equation}\label{A}
 2\left\lceil\frac{d_1}{4^i}\right\rceil
 =\left\lceil\frac{2d_1}{4^i}\right\rceil,
 \qquad 0\leq i\leq\nu-1.
\end{equation}
Using \eqref{A} and \(b\leq\nu\), we obtain
\begin{align*}
 n=3n_1
 &=3\tau_3+3\sum_{i=0}^{b-1}\left\lceil\frac{d_1}{4^i}\right\rceil\\
 &=3\tau_3+\sum_{i=0}^{b-1}\left(
 \left\lceil\frac{2d_1}{4^i}\right\rceil+
 \left\lceil\frac{d_1}{4^i}\right\rceil\right)\\
 &=3\tau_3+\sum_{j=0}^{2k_1-2\tau_3-1}\left\lceil\frac{2d_1}{2^j}\right\rceil.
\end{align*}
Since \(0\leq\tau_3\leq k_1-1\), this is an admissible term in \eqref{gri}; therefore the bound holds with equality. When \(k_1\leq\nu\), we have \(\tau_3=0\), which gives equality in the classical binary Griesmer bound; in this case, the locality constraint gives no improvement over the classical bound.
\end{proof}

Below, we provide several examples to illustrate Theorem~\ref{t42}.
\begin{Example}\label{d=7,8}
    For $k_1=\nu=2$, there exist Griesmer codes $\C_1$ with parameters $[9,2,7]$ and $\C_2$ with parameters $[10,2,8]$ over $\F_4$ in Code Tables \cite{codebase}. Construction \ref{constr} yields binary LRCs $\C_1'$ and $\C_2'$ with parameters $[27,4,14;2]$ and $[30,4,16;2]$, respectively. Both codes attain \eqref{gri} at $\tau=0$. 
\end{Example}

\begin{Example}
Take $k_1=4$ and $\nu=2$.
Code Tables \cite{codebase} lists quaternary Griesmer
codes $\C_3$ and $\C_4$ with parameters $[16,4,11]$
and $[17,4,12]$, respectively.
Using these outer codes in Construction~\ref{constr}
yields binary LRCs $\C_3'$ and $\C_4'$ with parameters
$[48,8,22;2]$ and $[51,8,24;2]$, respectively.
Both codes attain the Griesmer-type bound~\eqref{gri}
at $\tau=k_1-\nu=2$.
\end{Example}

Several known families of Griesmer codes have minimum distances that satisfy the conditions of Theorem~\ref{t42}.

\begin{Example}[The generalized MacDonald codes \cite{Dodunekov}]\label{ex:macdonald}
Let $s\geq2$ and $1\leq u\leq s-1$ be integers. Let $S_{s,4}$
be an $s\times\bigl((4^s-1)/3\bigr)$ matrix whose columns correspond to all distinct points in $\mathrm{PG}(s-1,4)$.
For $\lambda\in\{1,2,3,4\}$, construct a matrix $G$ by placing
$\lambda$ copies of $S_{s,4}$ side by side and deleting from one
copy the columns representing the points of a fixed
$(u-1)$-dimensional projective subspace.
The generalized MacDonald code with generator matrix $G$
has parameters
\[
 \left[n_1=\frac{\lambda(4^s-1)-(4^u-1)}{3},\quad
 k_1=s,\quad d_1=\lambda4^{s-1}-4^{u-1}\right].
\]
Its length is the Griesmer sum
$\sum_{i=0}^{s-1}\lceil d_1/4^i\rceil
=\lambda(4^s-1)/3-(4^u-1)/3$.
Using this code as the outer code in Construction~\ref{constr},
we obtain a binary LRC with parameters
\[
 [n=\lambda(2^{2s}-1)-(2^{2u}-1),\quad
 k=2s,\quad d=\lambda2^{2s-1}-2^{2u-1};\ r=2].
\]
For $\lambda\in\{2,3,4\}$, Theorem~\ref{t42} applies with
$\nu=s=k_1$, giving equality in~\eqref{gri} at the classical
term $\tau=0$.

The MacDonald family $\lambda=1$ \cite{Patel} gives equality
at $\tau=1$. Indeed, for $s\geq3$, take $\nu=s-1$.
If $u\leq s-2$, then $d_1=4\cdot4^{s-2}-4^{u-1}$, so we may
take $a=4$ and $\mu=4^{u-1}$. If $u=s-1$, then
$d_1=3\cdot4^{s-2}$, so $a=3$ and $\mu=0$.
In both cases, Theorem~\ref{t42} gives $\tau=k_1-\nu=1$.
For $s=2$ and $u=1$, the outer code has parameters $[4,2,3]$
and Theorem~\ref{thm:mds} gives $\tau=k_1-1=1$.
\end{Example}

\begin{Example}[The Solomon and Stiffler codes \cite{SOLOMON}]
Let $S_{t,4}$ be a $t \times (4^t-1)/3$ matrix whose columns correspond to all distinct points in $\mathrm{PG}(t-1,4)$.
Let $F = \bigcup_{i=1}^{h} \mathcal{U}_i$ be a union of disjoint projective subspaces, where each $\mathcal{U}_i$ has dimension $u_i-1$ with $t > u_1 \geq u_2 \geq \dots \geq u_h \geq 1$, and at most three of the subspaces have the same dimension. 
The Solomon--Stiffler codes over $\mathbb{F}_4$ with generator matrix $G$, which is obtained by deleting from $S_{t,4}$ all columns corresponding to points in $F$, have parameters $[ n_1 = \frac{4^t-1}{3} - \sum_{i=1}^{h} \frac{4^{u_i}-1}{3},\; k_1 = t,\; d_1 = 4^{t-1} - \sum_{i=1}^{h} 4^{u_i-1}]$ and attain the Griesmer bound.
Using these Solomon--Stiffler codes as outer codes in Construction~\ref{constr}, we obtain binary LRCs with parameters $[ n=2^{2t}-1-\sum_{i=1}^{h}\left(2^{2u_i}-1\right),\, k=2t,\,d=2^{2t-1}-\sum_{i=1}^{h}2^{2u_i-1};\,r=2]$. 

For $t\geq3$, Theorem~\ref{t42} applies whenever the $u_i$ are
distinct and $u_1\leq t-2$. Indeed, we may take $\nu=t-1$,
$a=4$, and $\mu=\sum_{i=1}^{h}4^{u_i-1}$, whose base-$4$
digits are all zero or one. The resulting binary LRC therefore
attains~\eqref{gri} at $\tau=k_1-\nu=1$.
For example, when $t\geq5$, we may delete the points of a plane
and a line disjoint from it in $\mathrm{PG}(t-1,4)$. This gives
$u_1=3$, $u_2=2$, and $d_1=4^{t-1}-4^2-4$.
\end{Example}

\section{Construction of optimal binary LRCs attaining the sphere-packing-like bound}
In this section, we construct perfect binary LRCs
from perfect linear codes over $\F_4$.
We also give a sufficient condition under which
a $k$-optimal outer code yields a $k$-optimal binary
LRC with respect to the sphere-packing-like bound.
 
The classical sphere-packing bound for $[n,k,d]$ linear codes over $\F_4$ is
\begin{equation}\label{cp}
    4^k\leq \frac{4^n}{\mathcal{O}_d},
\end{equation}
where $\mathcal{O}_d=\sum_{i=0}^{\lfloor \frac{d-1}{2}\rfloor}\binom{n}{i}3^i$.

Since $k$ must be an integer, taking logarithms of \eqref{cp} leads to the following bound:
\begin{equation}\label{cp1}
    k\leq  n-\left \lceil \log_4\mathcal{O}_d \right \rceil.
\end{equation}

\begin{Definition}
     An $[n,k,d]$ linear code over $\F_4$ is called a perfect code if it attains equality in (\ref{cp}); it is called a $k$-optimal code with respect to the sphere-packing bound  if its dimension attains the upper bound in \eqref{cp1}.
\end{Definition}

For integers $a,b$ with $0\le a<b$, set $\binom{a}{b}=0$.

For any $[n,k,d;r]$ LRC $\C$ with disjoint local repair groups, Wang \emph{et al.}~\cite{Wang} proposed a sphere-packing approach for deriving an upper bound. Specifically, for an $[n=3\ell,k,d=2t+2;r=2]$ binary LRC with disjoint local repair groups, the sphere-packing bound is as follows:
  \begin{equation}\label{perf}
        2^k\leq \frac{2^{\frac{2n}{3}}}{\Omega_d},
    \end{equation}
    where $\Omega_d=\sum_{0\leq i_1+i_2+\dots+i_{\ell}\leq \lfloor \frac{d-1}{4}\rfloor\ }\prod_{j=1}^{\ell}\binom{3}{2i_j}$.

Since $k$ must be an integer, taking logarithms of \eqref{perf} leads to the following bound:
\begin{equation}\label{perf1}
    k\leq   \frac{2n}{3}-\left \lceil \log_2\Omega_d \right \rceil. 
\end{equation}

\begin{Definition}[\cite{Fang}]\label{def:perfect}
An \([n=3\ell,k,d=2t+2;r=2]\) binary LRC with disjoint local repair groups is called perfect if it attains equality in \eqref{perf}. It is called \(k\)-optimal with respect to the sphere-packing-like bound if its dimension attains the upper bound in \eqref{perf1}.
\end{Definition}

Every perfect LRC is \(k\)-optimal with respect to the sphere-packing-like
bound. The converse can fail because an LRC may attain the rounded
dimension bound~\eqref{perf1} with strict inequality in~\eqref{perf}.
 
\begin{Theorem}\label{conper}
In Construction~\ref{constr}, if the outer code $\C$ is an $[n_1,k_1,d_1]$ perfect code over $\F_4$, then the concatenated code $\C'$ is an $[n=3n_1,k=2k_1,d=2d_1;r=2]$ binary perfect LRC.
\end{Theorem}

\begin{proof}
It is easy to verify that $2^k = 4^{k_1}$ and $2^{\frac{2n}{3}} = 4^{n_1}$.

\begin{equation}\label{eq:omega}
\begin{aligned}
    \Omega_d = \Omega_{2d_1}
    &= \sum_{0 \leq i_1 + i_2 + \dots + i_\ell \leq \lfloor \frac{2d_1-1}{4} \rfloor} \; \prod_{j=1}^{\ell} \binom{3}{2i_j} \\
    &\stackrel{(\mathbf{B_1})}{=} \sum_{0 \leq i_1 + i_2 + \dots + i_\ell \leq \lfloor \frac{d_1-1}{2} \rfloor} \; \prod_{j=1}^{\ell} \binom{3}{2i_j} \\
    &\stackrel{(\mathbf{B_2})}{=} \sum_{i=0}^{\lfloor \frac{d_1-1}{2} \rfloor} \binom{n_1}{i} 3^i = \mathcal{O}_{d_1},
\end{aligned}
\end{equation}
where ($\mathbf{B}_1$) follows from $\lfloor(2d_1-1)/4\rfloor=\lfloor(d_1-1)/2\rfloor$ for every integer $d_1$,   and the equality ($\mathbf{B}_2$) holds since  $n_1 = \ell$,  $\binom{3}{0}=1$, $\binom{3}{2}=3$ and $\binom{3}{2i_j}=0$ for all $i_j\geq 2$.

Therefore, if $\mathcal{C}$ is a perfect code, then $\mathcal{C}'$ is a binary perfect LRC.
\end{proof}

The Hamming codes are the well-known linear perfect codes over \(\F_4\).

\begin{Corollary}\label{cor:hamming-weight}
Let \(t\geq2\) be an integer, and set \(N=(4^t-1)/3\) and \(B=4^{t-1}\). In Theorem~\ref{conper}, take the \([N,N-t,3]\) Hamming code over \(\F_4\) as the outer code \(\C\). Then the concatenated code \(\C'\) is a binary perfect LRC with parameters \([3N,2(N-t),6;2]\), whose weight distribution is given by
\begin{equation}
\begin{aligned}
A_0(\C') &= 1,\\
A_{2j}(\C') &=
\frac{1}{4^t}\left[
3^j\binom{N}{j}
+(4^t-1)\sum_{a=0}^{j}(-1)^a 3^{j-a}
\binom{B}{a}\binom{N-B}{j-a}
\right], \quad 1\leq j\leq N,\\
A_i(\C') &= 0 \quad \text{otherwise}.
\end{aligned}
\end{equation}
\end{Corollary}
\begin{proof}
The dual of \(\C\) is the simplex code of dimension \(t\) over \(\F_4\), whose \(4^t-1\) nonzero codewords all have weight \(B\). The MacWilliams identity in Krawtchouk form \cite{MacWilliams} gives, for every $0\leq j\leq N$,
\begin{equation}
\begin{aligned}
 A_j(\C)&=\frac{\mathcal K_j(0;N;4)+(4^t-1)\mathcal K_j(B;N;4)}{4^t}\\
 &=\frac{1}{4^t}\left[3^j\binom Nj+(4^t-1)
 \sum_{a=0}^j(-1)^a3^{j-a}\binom Ba\binom{N-B}{j-a}\right],
\end{aligned}
\end{equation}
where $\mathcal K_j(i;N;q)=\sum_{a=0}^j(-1)^a(q-1)^{j-a}\binom ia\binom{N-i}{j-a}$. Corollary~\ref{weight} then gives the claimed weight distribution of $\C'$.
\end{proof}

 \begin{Example}\label{hanmming}
There exists a Hamming  code $\mathcal{C}$ over $\mathbb{F}_4$ with parameters $[5, 3, 3]$, whose parity-check matrix is given by
\begin{equation*}
     \begin{pmatrix}
     1& 0 &1 & 1& 1\\
     0&1&1& w&w^2\\
     \end{pmatrix},
\end{equation*}
where $w$ is a primitive element of $\F_4$.
In Construction \ref{constr}, taking the Hamming  code $\mathcal{C}$ as the outer code yields a binary LRC $\mathcal{C}'$ with parameters $[15, 6, 6;2]$, whose parity-check matrix is given by
\begin{equation*}
\left(
\begin{array}{ccc|ccc|ccc|ccc|ccc}
1&1&1 & & & & & & & & & & & & \\
& & & 1&1&1 & & & & & & & & & \\
& & & & & & 1&1&1 & & & & & & \\
& & & & & & & & & 1&1&1 & & & \\
& & & & & & & & & & & & 1&1&1 \\
\hline
0&1&0& 0&0&0& 0&1&0& 0&1&0& 0&1&0 \\
0&0&1& 0&0&0& 0&0&1& 0&0&1& 0&0&1 \\
0&0&0& 0&1&0& 0&1&0& 0&0&1& 0&1&1 \\
0&0&0& 0&0&1& 0&0&1& 0&1&1& 0&1&0 \\
\end{array}
\right).
\end{equation*}

The weight distributions of the codes $\mathcal{C}$ and $\mathcal{C}'$ satisfy
$A_0(\mathcal{C}')=A_0(\mathcal{C})=1$, $A_6(\mathcal{C}')=A_3(\mathcal{C})=30$, $A_{8}(\mathcal{C}')=A_4(\mathcal{C})=15$, $A_{10}(\mathcal{C}')=A_5(\mathcal{C})=18$.
\end{Example}

\begin{Corollary}\label{cor:sphere-window}
 In Construction~\ref{constr}, let the outer code $\C$ be a $k$-optimal $[n_1,k_1,d_1]$ code over $\F_4$ with respect to the sphere-packing bound, and suppose that $2^{2m-1}<\mathcal{O}_{d_1}\leq2^{2m}$ for some positive integer $m$. Then the concatenated code $\C'$ is a binary $[n=3n_1,k=2k_1,d=2d_1;r=2]$ $k$-optimal LRC with respect to the sphere-packing-like bound.
\end{Corollary}
\begin{proof}
The proof of Theorem~\ref{conper} gives
$\mathcal O_{d_1}=\Omega_d$.
The assumed inequality
$2^{2m-1}<\mathcal O_{d_1}\leq2^{2m}$ implies
\[
2\left\lceil\log_4\mathcal O_{d_1}\right\rceil
=2m
=\left\lceil\log_2\Omega_d\right\rceil.
\]
Since the outer code is $k$-optimal, we obtain
\[
k=2k_1
=2n_1-2\left\lceil\log_4\mathcal O_{d_1}\right\rceil
=\frac{2n}{3}
-\left\lceil\log_2\Omega_d\right\rceil.
\]
Thus, $\C'$ attains~\eqref{perf1}.
\end{proof}
\begin{Remark}\label{rem:sphere-window}
    If \(2^{2m}<\mathcal O_{d_1}\leq 2^{2m+1}\) for some integer \(m\geq0\),
then a \(k\)-optimal outer code with respect to the sphere-packing
bound yields a binary LRC whose dimension is one less than the
upper bound in~\eqref{perf1}.
\end{Remark}
\begin{Example}
    In Code Tables \cite{codebase}, there exists a $[43,36,5]$ cyclic code $\C$ over $\F_4$ generated by $g(x)=x^7 + wx^5+x^4+x^3 +w^2x^2+1$,  where $w$ is a primitive element of $\F_4$ satisfying $w^2=w+1$.  This code is  a $k$-optimal linear code with respect to the sphere-packing bound, as $2^{13}< \sum_{i=0}^{2}\binom{43}{i}3^i <2^{14}$. In Construction \ref{constr}, taking $\mathcal{C}$ as the outer code yields a binary $k$-optimal  LRC $\mathcal{C}'$ with respect to the sphere-packing-like bound, having parameters $[129, 72, 10;2]$.

\end{Example}

The same concatenation framework also yields binary
LRCs with locality $r>2$, using Hamming codes over
$\F_{2^r}$ as outer codes.

Let \(r>2\) and \(m\geq2\) be integers, and set \(Q=2^r\) and \(N=(Q^m-1)/(Q-1)\). As in Section~3, take the \([N,N-m,3]\) Hamming code over \(\F_Q\) as the outer code and concatenate it with the binary \([r+1,r,2]\) inner code. The resulting binary LRC has \(N\) disjoint repair groups of size \(r+1\) and minimum distance at least six.

To attain \(d=6\), choose the inner encoding \(f\) so that \(f(1)\) has weight two, and choose the parity-check matrix of the outer code so that it contains the columns \(\boldsymbol a\), \(\boldsymbol b\), and \(\boldsymbol a+\boldsymbol b\), where \(\boldsymbol a,\boldsymbol b\in\F_Q^m\) are linearly independent. Since \(\F_Q\) has characteristic two, these columns sum to zero; hence the outer code contains the codeword that is \(1\) in these three positions and \(0\) elsewhere. Its concatenated image has weight six. Together with the lower bound \(d\geq6\), this gives a binary LRC with parameters $[(r+1)N,\; r(N-m),\;6;\; r]$.

Since the resulting binary LRC has minimum distance six and
$N$ disjoint repair groups of size $r+1$, the sphere-packing-like
bound in \cite[Theorem~3]{Wang} gives
\begin{equation}\label{perf-general}
k\leq rN-\left\lceil
\log_2\left(1+N\binom{r+1}{2}\right)
\right\rceil.
\end{equation}
For \(r=3,4\), one has
\[
2^{rm-1}<1+N\binom{r+1}{2}<2^{rm},
\]
so the  dimension \(r(N-m)\) attains the bound \eqref{perf-general}. The parameters of both families of binary LRCs also appear in
\cite[Theorem~8]{Fang}. For \(m=2\), the corresponding binary LRCs with parameters \([36,21,6;3]\) and \([85,60,6;4]\) also appear in \cite[Construction~4.18]{Ma} and \cite[Construction~1]{Wang}, respectively.

For \(r\geq5\), \(\binom{r+1}{2}\leq(2^r-2)/2\), and consequently
\[
 1+N\binom{r+1}{2}\leq2^{rm-1}.
\]
Thus, for \(r\ge5\), using Hamming codes as outer codes still yields a valid distance-six construction. However, since the resulting binary LRCs do not attain the rounded bound \eqref{perf-general}, this argument does not establish \(k\)-optimality. For example, taking \(r=5\) and \(m=2\) yields a binary \([198,155,6;5]\) LRC, whereas \eqref{perf-general} allows \(k\le156\). Thus, the dimension of this LRC is one less than the bound \eqref{perf-general}.

\section{Construction of optimal binary LRCs attaining the Johnson-like bound}
In this section, we sharpen the Johnson-like bound
of Ma and Ge for binary LRCs with locality two
and disjoint repair groups.
Using quaternary outer codes, we construct nearly
perfect and $k$-optimal binary LRCs with respect
to the refined bounds.
 
For an $[n,k,d]$ linear code over $\F_q$ with even $d$, the classical Johnson bound \cite[Theorem 2.3.8]{Huffman} is
\begin{equation}\label{CJohnson}
   q^k\leq \frac{q^n}{\mathcal{O}_d+\frac{\binom{n}{d/2}(q-1)^{d/2}}{A_q(n,d,d/2)}},    
\end{equation}
where $\mathcal{O}_d=\sum_{i=0}^{\lfloor \frac{d-1}{2}\rfloor}\binom{n}{i}(q-1)^i$ and $A_q(n,d,d/2)$ is the maximum number of codewords in a constant-weight code over $\F_q$ of length $n$ and minimum distance at least $d$ whose codewords have weight $d/2$.

In fact, since each codeword has weight $d/2$ and the minimum distance is at least $d$, the supports of any two distinct codewords must be disjoint, which immediately gives (see \cite[Section V]{Fu98}) \begin{equation}\label{Joh2}
     A_q(n,d,d/2)= \lfloor \frac{n}{d/2}\rfloor.
\end{equation}

The classical Johnson bound can be reformulated as
\begin{equation}\label{Johnson}
   q^k\leq \frac{q^n}{\mathcal{O}_d+\frac{\binom{n}{d/2}(q-1)^{d/2}}{\lfloor \frac{2n}{d}\rfloor}}\triangleq \frac{q^n}{\mathcal{O}_d'}.  
\end{equation}

Since $k$ must be an integer, taking logarithms of \eqref{Johnson}  leads to the following bound:
\begin{equation}\label{cj1}
    k\leq   n-\left \lceil \log_q\mathcal{O}_d' \right \rceil. 
\end{equation}

\begin{Definition}
     An $[n,k,d]$ linear code over $\F_q$, where $d$ is even, is called a nearly perfect code if it attains equality in (\ref{Johnson}); it is called a $k$-optimal code with respect to the Johnson  bound if its dimension attains the upper bound in (\ref{cj1}).
\end{Definition}

 Ma and Ge \cite[Theorem~3.1(2)]{Ma} used a Johnson-type argument to obtain an explicit dimension bound for binary LRCs with disjoint local repair groups. For an $[n = 3\ell, k, d = 2t + 2; r = 2]$ binary LRC with $t + 1$ even, their bound is
\begin{equation}\label{npbb}
2^k\leq \frac{2^{\frac{2n}{3}}}{\Omega_d+\frac{\sum_{ i_1+i_2+\dots+i_{\ell}= \lfloor \frac{d}{4}\rfloor\ }\prod_{j=1}^{\ell}\binom{3}{2i_j}}{\lfloor \frac{2n}{d}\rfloor}}=\frac{2^{\frac{2n}{3}}}{\Omega_d+\frac{\binom{\ell}{d/4}3^{d/4}}{\lfloor \frac{2n}{d}\rfloor}},
\end{equation}
where $\Omega_d= \sum_{i=0}^{\lfloor \frac{d-1}{4} \rfloor}\binom{\ell}{i}3^i$.

Since $k$ must be an integer, taking logarithms of \eqref{npbb} leads to the following bound:
\begin{equation}\label{last}
    k\leq   \frac{2}{3}n-\left \lceil \log_2(\Omega_d+\frac{\binom{\ell}{d/4}3^{d/4}}{\lfloor \frac{2n}{d}\rfloor}) \right \rceil. 
\end{equation}

\begin{Theorem}\label{gai}
    For an $[n = 3\ell, \; k, \;d = 2t + 2;\; r = 2]$ binary LRC with disjoint local repair groups, where $t + 1$ is even, the bounds \eqref{npbb} and \eqref{last} can be improved as follows: 
    \begin{equation}\label{npbbb}
2^k\leq \frac{2^{\frac{2n}{3}}}{\Omega_d+\frac{\binom{\ell}{d/4}3^{d/4}}{\lfloor \frac{4n}{3d}\rfloor}}\triangleq \frac{2^{\frac{2n}{3}}}{\Omega_d'},
\end{equation}
\begin{equation}\label{las}
    k\leq   \frac{2}{3}n-
    \left \lceil \log_2\Omega_d' \right \rceil.  
\end{equation}
\end{Theorem}

\begin{proof}
 Let $\C$ be an $[n=3\ell,k,d=2t+2;r=2]$ binary LRC with disjoint local repair groups.  The $\L$-cover $W$ of $\C$ is the linear space spanned by the first $\ell$ rows of $H$ (as given in Equation \eqref{H1}), and the $\L$-space $\V$ of $\C$ is the dual of $W$.  Since $\C^\perp \supseteq W$, we have $\C\subseteq \V$. Using the argument for deriving the Johnson bound \cite[Theorem 2.3.8]{Huffman}, Ma and Ge~\cite[proof of Theorem~3.1]{Ma} obtained the following inequality:
 \begin{equation}\label{npb}
     |\C|\cdot |B_{\V}(t)|+|\mathcal{N}| \leq |\V|,
 \end{equation}
where $B_{\V}(t) =\{\v \in \V \mid \operatorname{wt}(\v) \leq t\}$, so $|B_{\V}(t)|=\Omega_d$ and $\mathcal{N}\triangleq \{\v\in \V| d(\v, \C)=t+1\}$.

 When $t+1$ is even, let $\Gamma \triangleq\{(\c,\x)|d(\c,\x)=t+1, \c \in \C, \x \in \mathcal{N} \}$.   Ma and Ge~\cite[Equation~(10)]{Ma} showed that \begin{equation}\label{e610}
     |\Gamma|=|\C| \cdot A_{t+1}(\mathcal{V})=2^k \cdot\binom{\ell}{d/4}3^{d/4},
 \end{equation}
 where  $A_{t+1}(\mathcal{V})=|\{\v \in \V \mid \operatorname{wt}(\v) = t+1\}|.$
 On the other hand, fix $\x\in\mathcal{N}$ and let $\Gamma_{\x}\triangleq\{\c\in\C\mid d(\c,\x)=t+1\}$. Then
 \begin{equation}\label{611}
      |\Gamma|=\sum_{\x\in \mathcal{N}}|\Gamma_{\x}|.
 \end{equation}

 Let $\Lambda_{\x}\triangleq\{\c-\x |d(\c, \x) =t+1, \c\in \C\}$. 
 We have $|\Gamma_{\x}|=|\Lambda_{\x}|$, and $\Lambda_{\x}$ is a binary constant-weight code of length $n$, weight $t+1$, and minimum distance at least $2t+2$. Hence, any two distinct vectors in \(\Lambda_{\x}\) have disjoint supports, and therefore
$|\Gamma_{\x}|=|\Lambda_{\x}|\leq\lfloor\frac{n}{t+1}\rfloor$.
Ma and Ge \cite[Equation~(11)]{Ma} used this estimate to derive the counting bound~\eqref{npbb} and hence the dimension bound~\eqref{last}.

Since \(\C\subseteq\V\) and \(\x\in\V\), we have \(\Lambda_{\x}\subseteq\V\). We use this structure to sharpen \eqref{npbb} by counting the occupied repair groups. Recall that $\V$ is the dual
of the space spanned by the first $\ell$ rows of $H$. Thus every
vector $\v\in\V$ can be written as
$\v=(\v_1\mid\v_2\mid\dots\mid\v_\ell)$, where each block
$\v_i\in\F_2^3$ has weight zero or two. Since any two nonzero blocks in the same repair group have intersecting supports, two vectors with disjoint supports cannot occupy the same repair group. Every vector in \(\Lambda_{\x}\) has weight \(t+1\), and hence occupies exactly \((t+1)/2\) repair groups. Consequently,
\[
|\Gamma_{\x}|=|\Lambda_{\x}|
\leq\left\lfloor\frac{\ell}{(t+1)/2}\right\rfloor
=\left\lfloor\frac{4n}{3d}\right\rfloor.
\]
It follows from~\eqref{611} that
\begin{equation}\label{gama}
|\mathcal N|\geq\frac{|\Gamma|}{\lfloor4n/(3d)\rfloor}.
\end{equation}
Combining~\eqref{npb}, \eqref{e610}, and~\eqref{gama} gives
\eqref{npbbb}. Taking base-$2$ logarithms and using the integrality
of $k$ gives~\eqref{las}.
\end{proof}

\begin{Definition}\label{near}
An \([n=3\ell,k,d=2t+2;r=2]\) binary LRC with disjoint local repair groups, where \(t+1\) is even, is called nearly perfect if it attains equality in \eqref{npbbb}. It is called \(k\)-optimal with respect to the Johnson-like bound if its dimension attains the upper bound in \eqref{las}.
\end{Definition}

Every nearly perfect LRC is \(k\)-optimal with respect to the
Johnson-like bound. The converse fails, as Examples~\ref{ex:length33}
and~\ref{ex:length36} attain the dimension bound~\eqref{las}
with strict inequality in the counting bound~\eqref{npbbb}.

\begin{Theorem}\label{conjhson}
In Construction~\ref{constr}, let the outer code $\mathcal{C}$ be an $[n_1,k_1,d_1]$ nearly perfect code over $\mathbb{F}_4$, where $d_1$ is even. Then the concatenated code $\mathcal{C}'$ is a binary $[n=3n_1,k=2k_1,d=2d_1;r=2]$ nearly perfect LRC.
\end{Theorem}

\begin{proof}
    The proof of Theorem \ref{conper} gives $2^k = 4^{k_1}$, $2^{\frac{2n}{3}} = 4^{n_1}$, and $\mathcal{O}_{d_1}=\Omega_d$. Since $n_1=\ell$ and $d=2d_1$, we also have
    \[\binom{n_1}{d_1/2}3^{d_1/2}=\binom{\ell}{d/4}3^{d/4}, \quad \lfloor \frac{2n_1}{d_1}\rfloor=\lfloor \frac{4n}{3d}\rfloor.\]

Therefore, $\mathcal{O}_{d_1}'=\Omega_d'$. Since $\mathcal{C}$ is nearly perfect, $\mathcal{C}'$ is a binary nearly perfect LRC.
\end{proof}

\begin{Example}\label{hexacode}
    The $[n_1=6,k_1=3, d_1=4]$ hexacode $\C$  is a nearly perfect code over $\F_4$, and its parity-check matrix is as follows:
    \begin{equation*}
         \left(
         \begin{array}{cccccc}
         1& 0 &0 & 1& w&w\\
         0&1&0& w&1&w\\
         0 &0&1 &w &w &1
         \end{array}
         \right),
    \end{equation*}
where $w$ is the primitive element of $\F_4$. It is easy to verify that \[\mathcal{O}_4+\frac{\binom{6}{2}3^{2}}{\lfloor \frac{6}{2}\rfloor}=64=4^{3}.\]

In Construction \ref{constr}, taking the hexacode $\mathcal{C}$ as the outer code yields a  binary $[18, 6, 8;2]$ nearly perfect LRC $\mathcal{C}'$, whose parity-check matrix is given by
\begin{equation*}
\left(
         \begin{array}{ccc|ccc|ccc|ccc|ccc|ccc}
         1&1&1 & &&  &&& &&& &&& &&&\\
		&&& 1&1&1   &&&  &&& &&& \\
            &&& &&& 1&1&1  &&& &&& &&&\\
		&&& &&&  &&& 1&1&1 &&& &&&\\
        &&& &&&  &&&  &&& 1&1&1 &&&\\
        &&& &&&  &&&  &&&  &&& 1&1&1\\
            \hline
		 0&1&0& 0&0&0&  0&0&0  &0&1&0&  0&0&1&  0&0&1  \\
         0&0&1& 0&0&0& 0&0&0 &0&0&1&  0&1&1&  0&1&1  \\
         0&0&0& 0&1&0&  0&0&0 &0&0&1&  0&1&0&  0&0&1  \\
         0&0&0& 0&0&1& 0&0&0& 0&1&1&  0&0&1&  0&1&1\\
         0&0&0& 0&0&0& 0&1&0& 0&0&1&  0&0&1&  0&1&0  \\
         0&0&0& 0&0&0& 0&0&1& 0&1&1&  0&1&1&  0&0&1  \\
         \end{array}
         \right).
\end{equation*}

The weight distributions of $\C$ and $\C'$ are given by
$A_0(\C')=A_0(\C)=1$, $A_8(\C')=A_4(\C)=45$, and $A_{12}(\C')=A_6(\C)=18$.
\end{Example}

\begin{Corollary}\label{cor:johnson-window}
 In Construction \ref{constr}, let the outer code $\mathcal{C}$ be a $k$-optimal code with respect to the Johnson bound, having parameters $[n_1,k_1,d_1]$ over $\mathbb{F}_4$ satisfying $2^{2m-1}<\mathcal{O}_{d_1}'\leq 2^{2m}$, where $d_1$ is even and $m$ is a positive integer. Then the concatenated code $\mathcal{C}'$ is a  binary $[n=3n_1,k=2k_1,d=2d_1;r=2]$ $k$-optimal LRC with respect to the Johnson-like bound.
\end{Corollary}
\begin{proof}
The proof of Theorem~\ref{conjhson} gives
$\mathcal O_{d_1}'=\Omega_d'$.
The assumed inequality
$2^{2m-1}<\mathcal O_{d_1}'\leq2^{2m}$ implies
\[
2\left\lceil\log_4\mathcal O_{d_1}'\right\rceil
=2m
=\left\lceil\log_2\Omega_d'\right\rceil.
\]
Since the outer code is $k$-optimal, we obtain
\[
k=2k_1
=2n_1-2\left\lceil\log_4\mathcal O_{d_1}'\right\rceil
=\frac{2n}{3}
-\left\lceil\log_2\Omega_d'\right\rceil.
\]
Thus, $\C'$ attains~\eqref{las}.
\end{proof}
\begin{Remark}\label{rem:johnson-window}
    The exponent-window condition in Corollary~\ref{cor:johnson-window} is  essential. If \(2^{2m}<\mathcal{O}_{d_1}'\leq2^{2m+1}\), then a \(k\)-optimal outer code with respect to the Johnson bound yields a concatenated LRC whose dimension is one less than the bound \eqref{las}.
\end{Remark}

Below, we give several examples of binary LRCs obtained from Construction~\ref{constr} that are \(k\)-optimal with respect to the Johnson-like bound but are not nearly perfect.

\begin{Example}\label{ex:cap17}
Code Tables~\cite{codebase} lists a linear code over $\F_4$ with
parameters $[17,13,4]$
(see \href{https://www.codetables.de/BKLC/BKLC.php?q=4&n=17&k=13}{Code Tables: $q=4$, $n=17$, $k=13$}).
Using this code as the outer code in Construction~\ref{constr},
we obtain a binary $[51,26,8;2]$ LRC. Here $\ell=17$ and $d=8$ give
\[
\Omega_8=52, \qquad \Omega'_8=52+\frac{1224}{8}=205.
\]
The refined bound~\eqref{las} gives
$k\leq34-\lceil\log_2 205\rceil=26$.
Since $52+1224/12=154$, the Ma--Ge bound~\eqref{last} also gives
$k\leq34-\lceil\log_2 154\rceil=26$.
Thus the constructed binary LRC attains both dimension bounds \eqref{last} and \eqref{las}.
\end{Example}

The next two examples attain the refined bound~\eqref{las} for parameters at which this bound is strictly tighter than \eqref{last}.

\begin{Example}\label{ex:length33}
Code Tables~\cite{codebase} lists a linear code over $\F_4$ with
parameters $[11,7,4]$
(see \href{https://www.codetables.de/BKLC/BKLC.php?q=4&n=11&k=7}{Code Tables: $q=4$, $n=11$, $k=7$}).
Using this code as the outer code in Construction~\ref{constr},
we obtain a binary $[33,14,8;2]$ LRC. Here $\ell=11$ and
$d=8$, so $t+1=d/2=4$ is even and Theorem~\ref{gai} applies. We have
\[
\Omega_8=1+11\cdot3=34, \qquad
\Omega'_8=34+\frac{\binom{11}{2}3^2}{5}
         =34+\frac{495}{5}=133.
\]
The refined bound~\eqref{las} therefore gives
\[
k\leq22-\lceil\log_2 133\rceil=14,
\]
whereas the Ma--Ge bound~\eqref{last} gives
\[
k\leq22-\left\lceil\log_2\left(34+\frac{495}{8}\right)\right\rceil=15.
\]
Thus the constructed binary LRC attains the tighter bound at $(n,d)=(33,8)$.
\end{Example}

\begin{Example}\label{ex:length36}
Code Tables~\cite{codebase} lists an optimal linear code over $\F_4$
with parameters $[12,6,6]$
(see \href{https://www.codetables.de/BKLC/BKLC.php?q=4&n=12&k=6}{Code Tables: $q=4$, $n=12$, $k=6$}).
Using it as the outer code in Construction~\ref{constr}, we obtain
a binary $[36,12,12;2]$ LRC. Here $\ell=12$ and $d=12$, so
$t+1=d/2=6$ is even and Theorem~\ref{gai} applies. We have
\[
\Omega_{12}=1+12\cdot3+\binom{12}{2}3^2=631, \qquad
\Omega'_{12}=631+\frac{\binom{12}{3}3^3}{4}
            =631+\frac{5940}{4}=2116.
\]
The refined bound~\eqref{las} gives
\[
k\leq24-\lceil\log_2 2116\rceil=12,
\]
which the constructed binary LRC attains. The Ma--Ge bound~\eqref{last} gives
\[
k\leq24-\left\lceil\log_2\left(631+\frac{5940}{6}\right)\right\rceil=13.
\]
\end{Example}

\section{Conclusion}\label{sec:conclusion}

In this paper, we construct binary LRCs with locality \(r=2\) and disjoint local repair groups via concatenated codes. Construction~\ref{constr} uses linear codes over \(\F_4\) as outer codes and the binary \([3,2,2]\) code as the inner code. The weight distribution of the outer code determines both the minimum distance and the weight distribution of the resulting LRC. Using this construction, the MDS and Griesmer codes over \(\F_4\) considered in Section~4 yield binary LRCs attaining the Griesmer-type bound~\eqref{gri}. We also construct perfect and \(k\)-optimal binary LRCs attaining \eqref{perf} and \eqref{perf1}, respectively. In addition, using Hamming codes over \(\F_4\) as outer codes, we derive an explicit formula for the weight distributions of the resulting perfect binary LRCs, whose parameters were already known.
We further sharpen the Johnson-like bound of Ma and Ge for binary LRCs with locality \(r=2\) and disjoint local repair groups. Under the same construction, the hexacode yields a nearly perfect binary \([18,6,8;2]\) LRC attaining the refined counting bound~\eqref{npbbb}. The binary \([33,14,8;2]\) and \([36,12,12;2]\) LRCs obtained from this construction attain the refined dimension bound~\eqref{las}, which is strictly tighter than the Ma--Ge bound~\eqref{last} at these code parameters. For \(r>2\), the choice of the isomorphism \(f\) in the general concatenated construction can affect the weights of the concatenated codewords. For \(r=3,4\), using Hamming codes over \(\F_{2^r}\) as outer codes yields \(k\)-optimal binary LRCs attaining the sphere-packing-like bound~\eqref{perf-general}. A direction for future research is to determine how to choose the outer codes and the isomorphism \(f\) jointly to construct \(k\)-optimal binary LRCs for general \(r>2\).

\section*{Acknowledgements}
The authors are grateful to the anonymous reviewers and the Associate Editor, Professor Gennian Ge, for their valuable suggestions and comments, which have greatly improved the paper.
\bibliography{sn-bibliography}

\end{document}